\documentclass[conference]{IEEEtran}
\IEEEoverridecommandlockouts
\usepackage{amsmath,epsfig}
\usepackage{times}  
\usepackage{helvet}  
\usepackage{courier}  
\usepackage[hyphens]{url}  
\usepackage{graphicx} 
\usepackage{amsmath}
\usepackage{mathtools}
\usepackage{breqn}
\usepackage{latexsym}
\usepackage{bbding}
\usepackage{amssymb}
\usepackage{booktabs}
\usepackage{epstopdf}
\usepackage{graphicx}
\usepackage{algorithm}
\usepackage{algpseudocode}
\usepackage{epstopdf}
\usepackage{amsfonts,amssymb} 
\usepackage{makecell}
\usepackage{multicol}
\usepackage{multirow}
\usepackage{tabularx}
\usepackage{setspace}
\usepackage{subfigure}
\floatname{algorithm}{Algorithm}

\usepackage{url}
\usepackage{caption} 

\usepackage{algorithm}
\usepackage{algorithmicx}  
\usepackage{booktabs}
\usepackage{newfloat}

\newtheorem{proof}{Proof}
\usepackage{listings}
\usepackage{hyperref}
\usepackage{enumitem}
\hypersetup{
    colorlinks=true,
    linkcolor=blue,
    filecolor=magenta,      
    urlcolor=cyan,
    pdftitle={Overleaf Example},
    pdfpagemode=FullScreen,
    }
\def\BibTeX{{\rm B\kern-.05em{\sc i\kern-.025em b}\kern-.08em
    T\kern-.1667em\lower.7ex\hbox{E}\kern-.125emX}}

\newcommand{\methodname}{{\tt{RC-TIM}}}
\newtheorem{theorem}{Theorem}
\newtheorem{lemma}[theorem]{Lemma}

\let\OLDthebibliography\thebibliography
\renewcommand\thebibliography[1]{
  \OLDthebibliography{#1}
  \setlength{\parskip}{0pt}
  \setlength{\itemsep}{0pt plus 0.3ex}
}

\begin{document}\sloppy

\def\x{{\mathbf x}}
\def\L{{\cal L}}

\title{A Renegotiable  contract-theoretic incentive mechanism for Federated learning}
%
\author{Xavier Tan$^{1}$, Xiaoli Tang$^{1}$, Han Yu$^{1}$ \\
$^1$College of Computing and Data Science, Nanyang Technological University, Singapore
}


\maketitle
\begin{abstract}
 Federated learning (FL) has gained prominence due to heightened concerns over data privacy. Privacy restrictions limit the visibility for data consumers (DCs) to accurately assess the capabilities and efforts of data owners (DOs). Thus, for open collaborative FL markets to thrive, effective incentive mechanisms are key as they can motivate data owners (DOs) to contribute to FL tasks. Contract theory is a useful technique for developing FL incentive mechanisms. Existing approaches generally assume that once the contract between a DC and a DO is signed, it remains unchanged until the FL task is finished. However, unforeseen circumstances might force a DO to be unable to fulfill the current contract, resulting in inefficient utilization of DCs' budgets. To address this limitation, we propose the \underline{R}enegotiable \underline{C}ontract-\underline{T}heoretic \underline{I}ncentive \underline{M}echanism (\methodname) for FL. Unlike previous approaches, it adapts to changes in DOs' behavior and budget constraints by supporting the renegotiation of contracts, providing flexible and dynamic incentives. Under \methodname{}, an FL system is more adaptive to unpredictable changes in the operating environment that can affect the quality of the service provided by DOs.
 Extensive experiments on three benchmark datasets demonstrate that \methodname{} significantly outperforms four state-of-the-art related methods, delivering up to 45.76\% increase in utility, on average.
\end{abstract}
\section{Introduction}
\label{sec:intro}
Federated learning (FL) has attracted significant research attention in recent years. To support open collaboration among data consumers (DCs) and data owners (DOs) under FL settings, it is important to provide DOs with suitable incentives. This is especially true for DOs whose primary tasks are not training machine learning models (e.g., hospitals, financial institutions) as their participation in FL often necessitates the diversion of valuable resources away from their primary tasks \cite{tan2023reputation}. 
In addition, due to privacy requirements, DCs cannot directly observe DOs' capabilities or effort, leading to information asymmetry that can result in inefficient incentivation outcomes \cite{smith2004contract}. 

\begin{figure}[!b]
    \centering
    \includegraphics[width=0.85\linewidth]{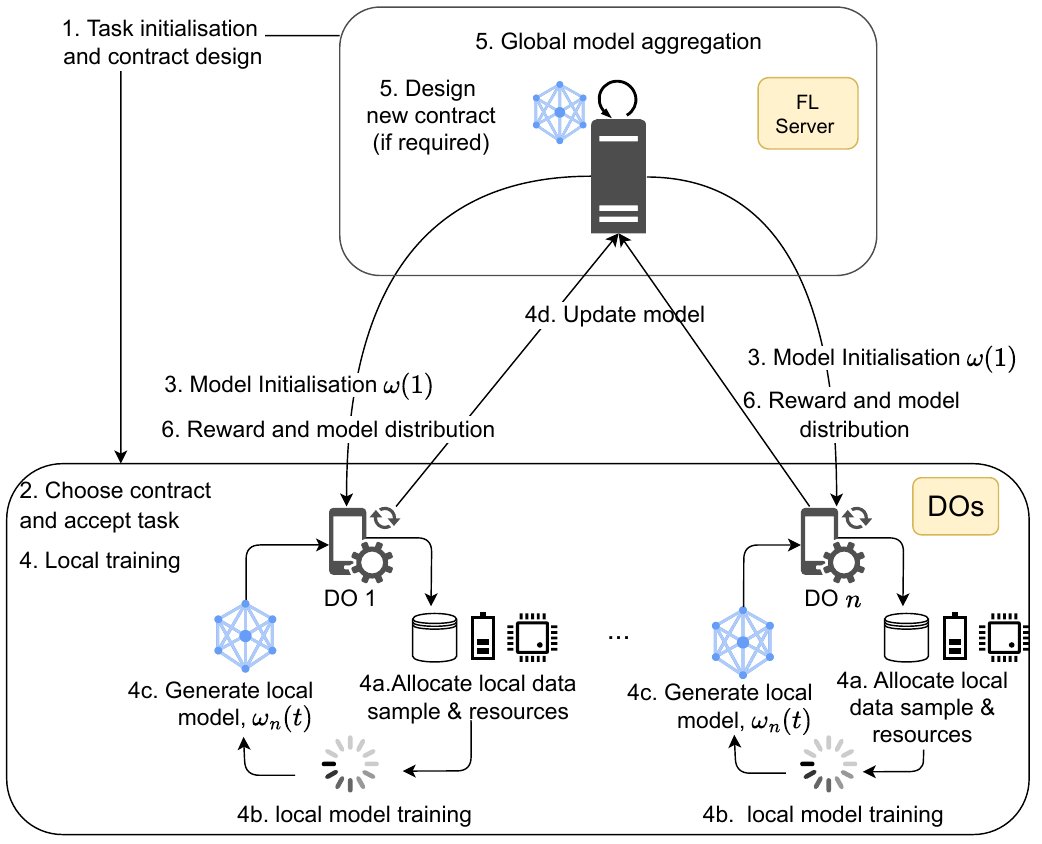}
    \caption{Illustration of the workflow of \methodname{}.}
    \label{fig:flowsummary}
\end{figure}

To this end, Contract Theory (CT) \cite{smith2004contract} has been adopted in the design of many FL incentive mechanisms \cite{kang2019incentive,lim2020dynamic,li2022contract}. The general approach is to examine how entities can reach optimal agreements when facing conflicting interests and asymmetrical information. In contract-based FL, DCs offer a set of contracts specifying the expected contribution levels and corresponding rewards, thereby allowing DOs to select from these contracts based on their respective types to join FL. The self-revealing nature of CT helps elicit optimal provisions even in the presence of information asymmetry. 

Existing contract-based FL faces practical limitations. Firstly, current methods assume full commitment from both DCs and DOs throughout the execution of a contract, which might not always be feasible. DOs might drop out due to network issues, limited battery life or intentional semi-honest behaviors, thereby resulting in partial contract fulfillment. Secondly, crafting optimal contracts requires DCs to accurately estimate the necessary resources (e.g., budget, time, data, compute power, communication bandwidth) to complete the training tasks. In real-world competitive environments, projects frequently exceed budgets or timelines, thereby necessitating to contract renegotiation or extension \cite{charette2005software}.
Moreover, FL tasks involve multiple training rounds, raising concerns about the sustained commitment of DOs \cite{bolton1990renegotiation}. The gap between expected and actual contributions can lead to uncertainty in rewards \cite{xu2024reciprocal}, resulting in suboptimal participation and budget over-estimation. Thus, it would benefit both the DC and DOs to periodically adjust their cost-reward expectations. Consequently, committing to a single contract for the entire duration of an FL task is neither feasible nor efficient.
Existing contract-based FL methods \cite{kang2019incentive,lim2020dynamic} are unable to accommodate dynamic contract renegotiation. 

To bridge this important gap, we propose the \underline{R}enegotiable \underline{C}ontract-\underline{T}heoretic \underline{I}ncentive \underline{M}echanism (\methodname) for FL. It is a two-stage renegotiable contract framework designed to be resilient to type misreporting or misrepresentation by considering real-time observations of DO behaviors and resource usage. In the first phase, an initial contract was created for the DC based on prior estimations. As more observations on DOs are gradually collected, the contract is revised via Bayesian updating based on the probability of a DO's type. This way, initial discrepancies in compensation adjustments can be addressed. To our knowledge, \methodname{} is the first renegotiable contract-based FL method, providing recourse regarding the distribution of DO types and reducing the strict reliance on precise distribution assumptions for optimal contract design. Extensive experiments on three benchmarking datasets demonstrate that \methodname{} significantly outperforms four state-of-the-art related methods, delivering improvements of up to 32\% compared to their average utility yield.

\section{Related works}
\label{related}
Contract theory provides a framework for understanding and designing agreements in situations where different parties have different information \cite{smith2004contract}. Contracts are designed to be individually rational and incentive compatible to motivate DOs to reveal their true capabilities and exert optimal effort \cite{smith2004contract}. However, the design of contract theory depends on knowing the distribution of DOs' types, which may not be apparent at the outset.


Kang et al. \cite{kang2019incentive} proposed to use Contract Theory and the reputation mechanism to address the challenge of selecting reliable DOs and incentivizing their participation in FL. The contract design involves specifying the resources DOs should contribute (i.e., data, compute), and the corresponding rewards they will receive. In \cite{lim2020dynamic,yu2024contract,yang2024asynchronous, cao2024federated}, Contract Theory is used in different FL settings to motivate DO participation while ensuring sustainable and privacy-preserving collaboration among them. Similarly in \cite{lim2020information}, Contract Theory is leveraged to incentivize privacy-preserving FL training while ensuring low service latency and age of information. The use of Contract Theory to attract DOs to participate in FL training by providing greater rewards to DOs with better data quality has been explored in \cite{li2022contract}. This is achieved by framing the optimization problem from the DOs' perspective instead of the DCs'. Taking a step further, \cite{lim2021towards} introduced a multidimensional contract aimed to maximize profits in FL systems. It takes into account multiple factors that influence overall profitability. 

Despite the success of Contract Theory in FL incentive mechanism design, existing research does not make provision for contract renegotiation, which makes them lack of the necessary flexibility to respond to changes in situations facing DCs and DOs in practice.
\methodname{} addresses this limitation via gradual adjustments to the incentive structure as more information becomes available. It ensures that the contract aligns with the actual type of the DO over time. 

\section{The Proposed Approach}
\subsection{Contract-based FL System Model}
We consider a scenario where a set of DOs $\mathcal{N} = \{1, \ldots, n, \ldots N\}$ are available to participate in an FL task hosted by a DC. The task aims to train a global model $\omega(\mathcal{T})$ over a duration of $\mathcal{T}$ global communication rounds or until the model reaches a pre-defined target generalization accuracy, whichever occurs first. A high-level overview of the process is illustrated in Figure \ref{fig:flowsummary}.
During each communication round $t \in [0, \mathcal{T}]$, the model undergoes training through the following five stages. 


\textbf{Stage 1: Task initialisation and contract design.} 
At this stage, the DC initiates a task to train the global model $\omega(t)$, by minimizing the loss function, with performance denoted by $\xi(\omega(t))$. The DC would classify the available DOs into $K$ categories, forming the set $\Theta = \{ \theta_1, \ldots, \theta_K\}$, based on their data sample size levels, sorted in ascending order: $ \theta_1 <  \ldots < \theta_k < \ldots \theta_K, \forall k \in \{1, \ldots, K\}$. A larger $\theta_K$ signifies that the DO possesses a large amount of available data, which in turn might expedite the FL training task. The DC provides contract $\Upsilon = (R_k, e_k)$ for category $\theta_k$, where $e_k = x_k \cdot d_k$ represents the effort required from a DO of type $k$, expressed as the product of the number of local epochs $x_k$ and the data sample size $d_k$.  $R_k$ is the reward for fulfilling the contract. This initial classification can evolve as more information becomes available through repeated interactions. In the presence of information asymmetry, the DC lacks direct knowledge of DO types and must infer the probability, $\rho_k$, that a DO belongs to type $k$ based on available observations, where $\sum^K_k \rho_k = 1$.

\textbf{Stage 2: DO contract selection and acceptance.}
After receiving the contract menu $\Upsilon$ from the DC, each DO selects and accepts the one that maximizes its utility. 

\textbf{Stage 3: Model initialisation.}
Once the contract is accepted, the DC transmits the current global model, $\omega(t)$, to the participating DOs.

\textbf{Stage 4: Local model training and update.}
Each DO, $n \in \mathcal{N}$, trains its local model, $\omega_n(t)$, over a designated number of local epochs, $x_n$, using its own dataset and computational resources to maximize local model accuracy. It then submits the model updates to the DC.

\textbf{Stage 5: Model aggregation and \methodname{}.}
After receiving DO model updates, the DC performs model aggregation using FL algorithms like FedAvg \cite{mcmahan2017communication} to update the global model $\omega(t)$. After multiple rounds of FL training, the DC can better ascertain the probability of DOs being associated with specific types and revise their respective optimal contract, $\Upsilon^*$, that is more suitable for both parties' interest. If both parties agree to the new terms, the revised contract $\Upsilon^*$ replaces the original one; otherwise, the original contract remains in effect.  

\textbf{Stage 6: Reward and model distribution.}
Upon contract fulfillment by the DO, the DC provides the agreed reward. 
The new updated model will also be distributed at this stage.

\subsection{Data Owners' Energy Cost}
The communication cost of DO $n$ for each training round $t$ with DC $f$ is given by:
\begin{equation}
 \label{eq:communication_cost_1}
    E^{comm}_{n}(t) =  T^{comm}_n p^{trans}_n  =  \frac{s_n(t)\cdot p^{trans}_n}{z_{n}},
\end{equation}
where $p^{trans}_n$ represents the transmission power of DO $n$, and $s_n(t)$  is the size of the local model, $\omega_n(t)$. We assume that the model size is constant across all participants engaged in the same FL task, as they are training the same global model.
The transmission rate, $z_{n}$, for DO $n$ given a transmission bandwidth $\beta$, is defined as: 
\begin{equation}
z_{n} = \beta \ln \left( 1 + \frac{\kappa_{n} p^{trans}_n}{\mathcal{H}_0} \right),
\end{equation}
where $\kappa_{n}$ is the channel gain of the link between DO $n$ and DC, while $\mathcal{H}_0$ represents the background noise. Thus, $E^{comm}_{n}(t)$ in Eq. \eqref{eq:communication_cost_1} could be rewritten as:
\begin{equation}
 \label{eq:communication_cost}
 E^{comm}_{n}(t) = \frac{s_n(t) \cdot p^{trans}_n}{ \beta \ln ( 1 + \frac{G_{n} p^{trans}_n}{\mathcal{H}_0} )}.
 \end{equation}
Each type of DO can allocate varying amounts of resources, which directly influences the number of local epochs they perform. The computational energy per global training round, denoted as $E^{comp}_n(x_n)$, is given by: $E^{comp}_n(x_n) = P^{cmp}_n \cdot T^{cmp}_n(x_n), $
where $P^{cmp}_n = \zeta_n \nu_n^2 \digamma_n$ is the computational power of DO $n$, and $T^{cmp}_n(x_n) = \frac{\mu_n d_n x_n}{\digamma_n}$ represents the computational time for local training. Here, $d_n$ refers to the amount of data that DO $n$ provides for training the FL model, and $\nu_n$ is the supply voltage required by DO $n$'s processor. The parameter $\mu_n$ denotes the total number of CPU cycles required to train a unit of data, while $\digamma_n$ represents the operating frequency of DO $n$'s CPU. Lastly, $\zeta_n$ is the effective load capacitance of DO $n$'s computational chip-set, and $x_n$ represents the number of local training rounds that DO $n$ performs during the global communication round. 

Then, the total cost for DO $n$ to participate in training FL server's, $f$, model during a single global communication round $t$ is the sum of computational and communication costs, expressed as:
\begin{equation}
 \label{eq:total_cost}
C^{total}_{n}(x_n,t)= \gamma_n \cdot(E^{cmp}_{n}(x_n) + E^{comm}_{n}(t)), 
\end{equation}
where $\gamma_n$ is the cost conversion factor that translates DO $n$'s energy consumption into a monetary or resource-based cost.

\subsection{Utility Functions}
The utility of a type $k$ DO, under the contract offered by the DC, quantifies the benefit or value that the DO derives from their participation. Taking into account the rewards received and the costs incurred, it is expressed as:
\begin{align}
\label{eq:utility_DO}
    U_k(e_k)= \theta_k R_k -  C^{total}_k(e_k),
\end{align}
where $e_k$ denotes the effort contributed by a DO of type $k$ as described above. Assuming rational behavior, DOs are likely to act in their own self-interest by selecting contracts that maximize their utility, seeking the highest rewards while minimizing associated costs. In other words, DOs will be more inclined to remain with the DC whom they have had established a strong rapport (i.e., recognized as high type). Therefore, a DO will aim to maximize its utility by selecting the most favorable contract available from the offered set, ensuring the best balance between reward and effort. The utility-maximizing behavior can be expressed as:
\begin{align}
    & \underset{(R_k, e_k)}{\max} U_k(e_k) = \rho_k b_k \biggr( \theta_k R_k -  C^{total}_k(e_k) \biggr),
\end{align}
where $\rho_k$ is the probability of the DO belonging to $k^{th}$ type. The binary variable $b_k$ indicates whether a DO of type $k$ is selected to participate in the FL task; if the DO is selected, $b_k = 1$, otherwise $b_k = 0$. It is important to note that, we assumed that each DO can participate in only one task at a time. The total utility for each DC, based on all the contracts offered, can be expressed as:
\begin{align}
\label{server_utility}
    & U = \sum^K_k  \rho_k \biggr ( Q [\xi(\omega_k) ] + \ln \biggr[ T_{max} \\
    & - \frac{\mu_k e_k}{\digamma_k} - T^{comm}_k \biggr] - \theta_k R_k \biggr ),  \nonumber
\end{align}
where $Q$ is the revenue conversion function based on model performance. 

\subsection{The \methodname{} Algorithm}
For DO $n$ to prefer the contract $k$ offered by DC over other available options, the contract must satisfy the following essential requirement.

\textbf{Definition 1.} (Incentive Compatibility) ensures that DOs are incentivized to truthfully disclose their private information and select contract that maximizes their expected utility,
\begin{align}
\label{IC}
&   \theta_k R_k - C^{total}_k(e_k) \geqslant   \theta_j R_j - C^{total}_k(e_j), \nonumber \\
& \forall k,j \in \{1, \ldots, K \} , k \neq j.
\end{align}

\textbf{Definition 2.} (Individual Rationality) ensures that DOs are not worse off by participating, meaning each DO will only contribute if their expected utility is non-negative
\begin{align}
\label{IR}
    U_k(e_k) =  \theta_k R_k - C^{total}_k(e_k)  \geqslant 0.
\end{align}

\textbf{Definition 3.} (Budget feasibility) requires that the total payment per global communication round must not exceed the maximum budget, $B^{max}$, predefined by the DC: $\sum^K_k \theta_k R_k \leq B^{max}$.

To maximize utility, the DC must balance total payments to DOs against the global model's performance. However, maximizing the objective defined in Eq. \eqref{server_utility} does not meet the requirements for convex optimization, making it difficult to derive an optimal solution directly. To address this, we first relax the incentive compatibility and individual rationality constraints, then iteratively verify solutions against Local Downward Incentive Compatibility (LDIC) and Local Upward Incentive Compatibility (LUIC) constraints \cite{celik2006mechanism}. In conjunction with the monotonicity constraints, both upward and downward incentive compatibility can be maintained \cite{yu2024contract}. As a result, the IC constraints can be considered to be the reduced IC problem. These conditions ensures that no agent (i.e., DOs) has the incentive to misreport their type (either under-reporting or over-reporting) to the best of their knowledge. Therefore, promoting honest participation and efficient outcomes. 
\subsection{Proofs for LDIC and LUIC}
\begin{lemma}
    If $\theta_1$'s IR constraint is satisfied, all IR constraint for other higher types can be reduced.
\end{lemma}
\begin{proof}
    It is given that,
    \begin{align}
        & \theta_k R_k - C^{total}_k(e_n) \geqslant \theta_k R_1 - C^{total}_1(e_1), \\
        & \theta_k R_1 - C^{total}_1(e_1) \geqslant \theta_1 R_1 - C^{total}_1(e_1), 
    \end{align}
    we can reduce the IR constraints to the following,
    \begin{align}
    \label{IRweak}
        \theta_1 R_1 - C^{total}_1(e_1) = 0.
    \end{align}
\end{proof}
\begin{lemma}
\label{lemmaIC}
Monotonicity: If $\theta_k \geqslant \theta_j$, then it must be true that $e_k \geqslant e_j $ such that inevitably $R_k \geqslant R_j$, where $k$ is a higher type than $j$, $\forall k,j \in \{1, \ldots , K\} $. 
\end{lemma}
\begin{proof}
According to definition of IC and Eq. \eqref{IC}, we know for sure that
\begin{align}
    & \theta_k R_k - C^{total}_k(e_k) \geqslant \theta_k R_j - C^{total}_j(e_j), \\
    & \theta_j R_k - C^{total}_j(e_j) \geqslant \theta_j R_k - C^{total}_k(e_k),
\end{align}
Combining the above two equation will yield us:$(\frac{1}{\theta_j} - \frac{1}{\theta_k})(e_k - e_j) \geqslant 0, $
and $(R_k - R_j) \geqslant C^{total}_k(e_k) - C^{total}_j(e_j)$ which can be further simplified into: 
\begin{align}
\label{ICworking}
    (R_k - R_j) \geqslant \mu_{k} \zeta_{k} {\nu}_{k}^2(e_{k}-e_{j}).
\end{align}
In other words, $(R_k \geqslant R_j)$ is true if and only if $e_k \geqslant e_j$, thus monotonicity must be held.
\end{proof}
\begin{lemma}
In conjunction with lemma \ref{lemmaIC}, the IC constraints can therefore be further reduced as a pair of LDIC and LUIC constraints, 
$\theta_k R_k - (\mu_k \zeta_k {\nu}_k  e_k) \geqslant \theta_k R_{k-1}- (\mu_{k-1} \zeta_{k-1} {\nu}_{k-1}^2  e_{k-1}), k \in \{2, \ldots, K \},  k \in \{2, \ldots, K \} $ and $\theta_k R^f_k - (\mu_k \zeta_k {\nu}^2_k  e_k) \geqslant \theta_k R_{k+1}^f- (\mu_{k+1} \zeta_{k+1} {\nu}_{k+1}^2  e_{k+1}), k \in \{1, \ldots, K-1 \}$ respectively. 

However, because of the monotonicity as aforementioned, the following can be deduced, $\theta_{k+1}(R_k - R_{k-1}) \geqslant \theta_{k}(R_k - R_{k-1}) \geqslant \mu_{k} \zeta_{k} {\nu}_{k}^2(e_{k}-e_{k-1})$. Thereafter, we can combine and simplify the pair of LDIC and LUIC to be:
\begin{align}
\label{reducedpair}
    & \theta_k R_k - (\mu_{k} \zeta_{k} {\nu}_{k}^2  e_{k}) \geqslant \theta_k R_{k-1}- (\mu_{k-1} \zeta_{k-1} {\nu}_{k-1}^2  e_{k-1}).
\end{align}
\end{lemma}

\begin{proof}
    From Eq. \eqref{IRweak}, it is in the interest of the DC to reduce $R_1$ as much as possible, such that they could maximise their utility yield (i.e., hiring DOs at cost price $\theta_1 R_1 - C^{total}_1(e_1) = 0$). This applies to LDIC as well, the DC would want to reduce reward value until $\theta_k R_k -C^{total}_1 = \theta_k R_{k-1} - C^{total}_1 (e_{k-1})$. This can be reformatted as $\theta_k R_k- \theta_k R_{k-1} = (\mu_n \zeta_n {\nu}^2_k  e_k-e_{k-1})$ combining with Eq. \eqref{ICworking} we can derive our reduced IC constraint Eq. \eqref{reducedpair}.
\end{proof}
Assuming that all DOs experience similar communication conditions across all communication rounds, meaning that for any $t$, we have $E^{comm}_1 = E^{comm}_2 = \ldots = E^{comm}_k$ for all $k \in K$, and that $\gamma_k$ remains constant, the utility function defined in Eq. \eqref{server_utility} is re-written as: 
\begin{align}
\label{obj1}
    \underset{(R_k, e_k)}{\max} U(e_k, R_k), 
\end{align}
subjected to:
\begin{align}
& \theta_1 R_1 - C^{total}_1(e_1) = 0, \label{eq:IRconstraint} \\ 
& \theta_k R_k - (\mu_k \zeta_k {\nu}_k^2  e_k) \label{eq:ICconstraint} \geqslant   \\
& \theta_k R_{k-1}- (\mu_{k-1} \zeta_{k-1} {\nu}_{k-1}^2  e_{k-1}), k \in \{2, \ldots, K \}, \nonumber \\
& \sum^K_{k=1}  \theta_k  R_k \leqslant B^{max}, \forall k \in K .\label{eq:Budgetconstraint} 
\end{align}
Eq. \eqref{eq:IRconstraint} - Eq. \eqref{eq:Budgetconstraint} represent the reduced versions of IC and IR requirements, as well as the budget constraints, respectively. By systematically incorporating the constraints as referenced from \cite{lim2020information}, we can derive $R_k$ as:
\begin{align}
\label{Rformula}
    R_k = \sum^K_{k=2} \frac{1}{\theta_k}\mu_{k} \zeta_{k} {\nu}_{k}^2(e_{k}-e_{k-1}) + \frac{1}{\theta_1}\biggr(C^{total}_1(e_1) \biggr).
\end{align}
In Eq. \eqref{Rformula}, the optimal reward $R_k$ is now dependent on the DO's effort $e_k$. This allows us to effectively solve Eq. \eqref{obj1} using a single variable.  In other words, we can iteratively determine the optimal contract reward $R_k(e_k)$ based on the set of feasible effort levels that each DO can provide and thereafter ensure that the solution satisfies the monotonicity constraint; whereby lower-type DOs exerting lesser effort receive a smaller reward compared to higher-type agents exerting greater effort. Substituting $R_k$ into the function for total expected rewards, $\sum^K_k \rho_k \theta_k R_k$, we can derive the total rewards required for DOs of types $k$ across the probability distribution $\rho_k$ :
\begin{align}
    \sum^K_k \rho_k \theta_k R_k = \sum^K_k X_k + \frac{C^{comm}_1}{\theta_1} \sum^K_k \theta_k \rho_k.
\end{align}
For $k < K$, we have:
\begin{align}
    \label{closedform}
    X_k = \mu_k \zeta_k {\nu}_k^2 e_k  + \mu_k \zeta_k {\nu}_k^2 e_k (\frac{1}{\theta_k}-\frac{1}{\theta_{k+1}}) \sum^K_{i=k+1}\theta_i \rho_i ,
\end{align}
if $k = K$, $ X_k = \mu_K \zeta_K {\nu}_K^2  e_K$. $X_k$ represents as a substitution variable to maintain the clarity and conciseness of the equations.
Using the closed-form solution, we can reduce the objective function to a single-variable problem. By applying convex optimization techniques, the optimal effort $\hat{e}_k$ and corresponding reward $\hat{R}^f_k$ can be derived. 
Thereafter, we derived that $\frac{\delta ^2U}{\delta e_k^2} \leqslant 0$, thus showing that it has a maximum point. Initially, we assumed a uniform distribution for the DO types. The number of local training epochs for each DO could then be determined as $x_k = \frac{e_k}{d_k}$. After $t$ communication rounds, the DC would have gained new observations, denoted as $\psi$, on the behaviors of participating DOs. These updated observations can be used to refine prior beliefs about the DO types. Using Bayes' theorem, the DC can update their type distribution as $\Pr(\theta=k|\psi) \propto \Pr(\psi)\Pr(\psi|\theta=k)$, incorporating the new information into their decision-making process. Then, the new type probability distribution becomes:  
\begin{align}
    \Pr(\theta=k | \psi) = \frac{\Pr(\psi|\theta=k)\Pr(\theta=k)}{\Pr(\psi)} ,
\end{align}
where $\Pr(\psi)$ can be calculated using the law of total probability: $\Pr(\psi) = \sum^K_{k=1} \Pr(\psi|\theta=k)\Pr(\theta=k)$. Thereafter, DC can propose the newly drafted contract, $\Upsilon^*$, to the participating DOs. To ensure budget feasibility, the DC checks that the model is converging as expected and that there is sufficient budget remaining at the point of re-contracting. Specifically, the DC verifies that the conditions $\sum^{T/a}_t \sum^K_k \theta_k R_k \leqslant \frac{B^{max}_f}{a}$ and $\omega(t) \leqslant \omega(t-1)$ held at round $t = \frac{T}{a}$, where $a$ represents a predefined partition of the total training rounds. The new contract will be designed according to Eq. \eqref{closedform}, ensuring it better aligns with the actual behaviors and capabilities observed, thereby enhancing both the fairness and effectiveness of the agreement. The pseudo-code for \methodname{} is exhibited as Algorithm \ref{alg:FedAlgo}.
\begin{algorithm}
\caption{\methodname}\label{alg:FedAlgo}
\begin{algorithmic}

\State \textbf{Initialize}: $\omega_f(0)$; $T_{max}$; training parameters.
\State Formulate the menu of contracts for DOs types based on prior probability distribution, $\Upsilon = (R_k(e_k) | \mathcal{P}_k)$;
\State Each DO $n$ choose preferred contract and participate in FL task;
\State DC publishes $\omega(0)$ to all participating DOs;
\While{ $t < T$ }
   \State When budget expenditure is available, perform FL training with FedAvg
    \If{$k =K$}
        \State $\hat{R}_k = \mu_k \zeta_k {\nu}_k^2  e_k $ ;
    \Else
        \State $\hat{R}_k = \mu_k \zeta_k {\nu}_k^2 e_k  + \mu_k \zeta_k {\nu}_k^2 e_k (\frac{1}{\theta_k}-\frac{1}{\theta_{k+1}})$ ;
    \EndIf
    \If{$t = \frac{T}{a}$ \& $\sum^{T/a}_t \sum^K_k R_k \leqslant \frac{B^{max}_f}{a}$ \& $\omega(t) \leqslant \omega(t-1) $}
        \State update $\psi$ and probability distribution;
        \State $\Pr(\theta=k | \psi) = \frac{\Pr(\psi|\theta=k)\Pr(\theta=k)}{\Pr(\psi)} $ ;
        \State Reformulate contract,  $\Upsilon^* = [ R^k(e_k)| (\mathcal{P}_k | \psi) ] $ ;
        \If{$n$ accepts new contract}
            \State $\Upsilon^*$ takes effect, and DC pays the new $\hat{R}^k$ to DO $n$ of type $k$.
        \Else 
            \State $\Upsilon$ is still effective.
        \EndIf
    \EndIf        
\EndWhile
\end{algorithmic}
\end{algorithm}

\section{Experimental Evaluation}
\label{exp}
In this section, we evaluate \methodname{} against four other state-of-the-art (SOTA) approaches based on three benchmarking datasets. 
\begin{table}[ht]
\centering
\caption{Experiment parameters.}
\label{tab:parameters}
\resizebox{1\linewidth}{!}{
\begin{tabular}{|r|r|r|r|r|l|}
\hline
\textbf{Parameters} & \textbf{Value} & \textbf{Parameters} &  \textbf{Value} & \textbf{Parameters} & \textbf{Value}\\
\hline
 Budget & $400$ & $T_{max}$ & $1500ms$ & Batch Size & $128$ \\ 
 Total Rounds & $50$ & Total DOs & $45$ & Learning rate & $0.01$ \\
 DO starting price & $0.2$ &  $\gamma_f$ & $0.003$ &  $E^{cmp}_n$ & $0.01$ \\
 Momentum & $0.9$ & Number of types & $10$ & $E^{comm}_n$ & $0.1$ \\
\hline
\end{tabular}
}
\end{table}
\subsection{Experimental Settings}
The model used for training on the CIFAR-10 dataset \cite{krizhevsky2009learning} consists of 1,006,206 parameters, structured with two convolutional neural network (CNN) blocks followed by three fully connected (FC) layers, including dropouts before the third layer. The model used for EMNIST \cite{cohen2017emnist} balanced dataset comprises 907,491 parameters and has similar architecture to CIFAR-10's, but features only one input channel and no dropout layers. 

The model for MNIST \cite{lecun_bottou_bengio_haffner_1998} has 21,840 parameters, consisting of two simple CNN and max-pooling blocks, followed by two smaller FC layers. Initially, the $45$ DOs are split into $10$ different types as $\{1, \ldots , 10\}$, with a uniform probability of $0.1$. For the revenue conversion function we set $Q(\xi(\omega_k) )=2 \cdot \xi(\omega_k) $. We tested our approach alongside four other SOTA methods in both IID and non-IID scenarios. For the \methodname{} approach, we assume that DOs initially have a uniform prior distribution of types. After 25 of 50 communication rounds, if the DC observes that a DO’s type deviates from expectations, it can update its belief on DO's type and propose a new contract, provided certain conditions are satisfied. Other hyperparameter settings are documented in Table \ref{tab:parameters}.

\begin{table}[ht]
\centering
\caption{Simulation Results.}
\label{tab:results}
\resizebox{0.8\linewidth}{!}{
\begin{tabular}{|r|r|c|c|}
\hline
\multicolumn{2}{|r|}{} & \textbf{IID} & \textbf{Non-IID}\\
\hline
\textbf{Method} & \textbf{Dataset} & \multicolumn{2}{c|}{\textbf{Utility} ($\times10^2$)}  \\
\hline
\multirow{3}{*}{\methodname{}} & MNIST &  $\mathbf{77.20}$ & $\mathbf{73.37}$ \\
                     & EMNIST-balanced &  $\mathbf{79.20}$ & $\mathbf{59.29}$ \\
                     & CIFAR10 &  $\mathbf{44.48}$ & $\mathbf{43.18}$ \\
\hline                     
\multirow{3}{*}{Contract} & MNIST &  $76.76$ & $70.20$ \\
                     & EMNIST-balanced &  $77.16$ & $58.66$ \\
                     & CIFAR10 &  $34.67$ & $38.70$ \\
\hline   
\multirow{3}{*}{GTG-SV} & MNIST &  $68.37$ & $56.47$ \\
                     & EMNIST-balanced &  $76.98$ & $58.31$ \\
                     & CIFAR10 &  $36.97$ & $28.58$ \\  
\hline                        
\multirow{3}{*}{OORT} & MNIST &  $67.66$ & $54.82$ \\
                     & EMNIST-balanced &  $73.44$ & $53.65$ \\
                     & CIFAR10 &  $35.05$ & $23.82$ \\    
\hline
\multirow{3}{*}{RRAFL} & MNIST &  $62.06$ & $53.16$ \\
                     & EMNIST-balanced &  $68.22$ & $56.16$ \\
                     & CIFAR10 &  $29.07$ & $31.04$ \\    
\hline
\end{tabular}
}
\label{tab:experiment}
\end{table}

\subsection{Comparison Baselines}
We compared our approach against four other state-of-the-art approaches that have demonstrated strong performance in the context of incentive-based client selection in FL. 
\begin{enumerate}
    \item \textbf{GTG-SV} \cite{liu2022gtg}: Shapley value (SV) is a method used to evaluate the contribution level of each participant, where rewards are distributed proportionately based on each participant's associated SV. To reduce the computational complexity of SV, GTG-SV provides a guided estimation method.    
    \item \textbf{OORT} \cite{lai2021oort}: This approach gradually and effectively selects participants, balancing the exploration-exploitation dilemma to identify high-performing contributors.
    \item \textbf{RRAFL} \cite{zhang2021incentive}: A reputation-aware incentive mechanism designed within a reverse auction framework, aiming to select the top $k$ number of reputable participants while adhering to a given budget constraint.
    \item \textbf{Contract}: This approach is based purely on an optimal contract, with rewards corresponding to Eq. \eqref{Rformula} and Eq. \eqref{closedform}.
\end{enumerate}

\subsection{Results and Analysis}
The experimental results are shown in Table \ref{tab:results}, with the best results in bold.
The table shows that \methodname{} outperforms other state-of-the-art approaches, especially in non-IID scenarios. To illustrate, on the MNIST non-iid dataset,  \methodname{} scored a total utility of $4.52\%$, $29.93\%$, $33.84\%$, $38.02\%$ higher than the basic contract-based approach, GTG-SV, OORT, RRAFL, respectively. While under a more challenging dataset, CIFAR10 under non-iid scenario, \methodname{} scored a total utility of $11.58\%$, $51.08\%$, $81.27\%$, $39.11\%$  higher than the other approaches in the same order.

\section{Conclusions and Future Work}
\label{conclusion}
In this paper, we introduced an incentive mechanism based on contract theory, \methodname{}, to align rewards with the capabilities and preferences of DOs. However, contract theory relies on assumptions about the DOs' types, which may be inaccurate due to issues such as misreporting, changes in DO behavior over time, or incorrect assumptions about the probability distribution of DOs' types. By leveraging on the fact that FL occurs over multiple communication rounds, we capitalize on the new information revealed during these rounds, enabling us to update the probability of a DO belonging to a specific type, and therefore reformulate a new contract menu. Extensive experimental results demonstrated that \methodname{} achieved superior utility results, especially in non-IID scenarios, compared to existing state-of-the-art methods. To the best of our knowledge, \methodname{} is the first renegotiable contract theory-based incentive mechanism designed for federated learning.

In subsequent research, we plan to improve \methodname{} by investigating more effective methods for updating the contract and developing a more dynamic approach to determine the optimal condition for the DC to initiate contract review for each individual DO.

\section*{Acknowledgements}
\thanks{This research/project is supported, in part, by the National Research Foundation, Singapore and DSO National Laboratories under the AI Singapore Programme (AISG Award No: AISG2-RP-2020-019), under Energy Research Test-Bed and Industry Partnership Funding Initiative, part of the Energy Grid 2.0 programme, and under DesCartes and the Campus for Research Excellence and Technological Enterprise (CREATE) programme; Alibaba Group through Alibaba Innovative Research (AIR) Program and Alibaba-NTU Singapore Joint Research Institute (JRI) (Alibaba-NTU-AIR2019B1), Nanyang Technological University, Singapore; the RIE 2020 Advanced Manufacturing and Engineering (AME) Programmatic Fund (No. A20G8b0102), Singapore; Nanyang Technological University, Nanyang Assistant Professorship (NAP); and Future Communications Research \& Development Programme (FCP-NTU-RG-2021-014).}

\bibliographystyle{IEEEbib}
\bibliography{main}

\end{document}